\documentclass{llncs}
\usepackage{amsmath}
\usepackage{graphicx}
\usepackage[caption=false]{subfig}
\usepackage[pagebackref=false,pdfstartview=FitH,bookmarksopen=false,colorlinks,linkcolor=black,citecolor=black]{hyperref}

\newcommand{\suc}[1]{\ensuremath{\text{succ}({#1})}}
\newcommand{\sqrtsig}{\ensuremath{\sqrt{\sigma}}}
\newcommand{\sqrtsiglg}{\ensuremath{\sqrtsig\lg(n/\sqrtsig)}}
\newcommand{\sqrtsiglga}{\ensuremath{\sqrt{\sigma_1}\lg(n_1/\sqrt{\sigma_1})}}
\newcommand{\sqrtsiglgb}{\ensuremath{\sqrt{\sigma_2}\lg(n_2/\sqrt{\sigma_2})}}
\newcommand{\sqrtsiglgc}{\ensuremath{\sqrt{\sigma_3}\lg(n_3/\sqrt{\sigma_3})}}

\begin{document}

\title{The Minimum Bends in a Polyline Drawing with Fixed Vertex Locations}
\author{Taylor Gordon}
\institute{Google}
\maketitle

\begin{abstract}
We consider embeddings of planar graphs in $R^2$ where vertices map to points
and edges map to polylines. We refer to such an embedding as a polyline drawing,
and ask how few bends are required to form such a drawing for an arbitrary
planar graph. It has long been known that even when the vertex locations are
completely fixed, a planar graph admits a polyline drawing where edges bend a
total of $O(n^2)$ times. Our results show that this number of bends is
optimal. In particular, we show that $\Omega(n^2)$ total bends is required to
form a polyline drawing on any set of fixed vertex locations for almost all
planar graphs. This result generalizes all previously known lower bounds, which
only applied to convex point sets, and settles 2 open problems.
\end{abstract}

\section{Introduction}

A \emph{polyline drawing} is an embedding of a planar graph in $R^2$ where
vertices map to points and edges map to polylines (piecewise linear curves). An
example of such a drawing is depicted in Figure~\ref{fig:polyline}.

\begin{figure}
\centering
\includegraphics[trim=0 0.5cm 0 0.5cm]{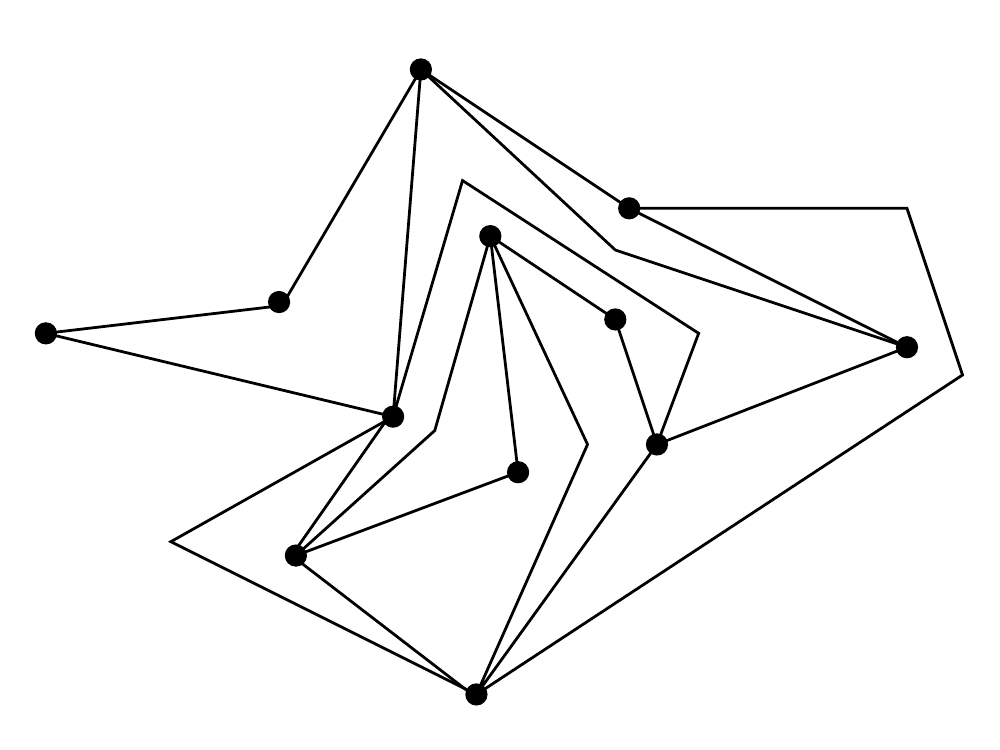}
\caption{A polyline drawing of a planar graph on 12 vertices.}
\label{fig:polyline}
\end{figure}

Applications of polyline drawings to \emph{VLSI circuit design} and
\emph{information visualization} (see \cite{battista98}, \cite{tamassia87}) have
inspired the development of algorithms for constructing polyline drawings of
arbitrary planar graphs. A particularly auspicious construction was shown in
\cite{schnyder90}, where all edges map to straight-line segments and all
vertices map into an $(n-2) \times (n-2)$ grid. 

What if we further restrict the locations to which vertices map? Do all planar
graphs admit polyline drawings if we fix the points to which vertices are
mapped? The answer is yes but with a commensurate increase in the complexity of
the polylines to which edges map. To what extent can this complexity be
minimized? One prevailing way to measure this complexity is by counting the
number of times the polylines \emph{bend} (see \cite{erten04},
\cite{kaufmann02}, \cite{tamassia87}, \cite{pach98}).

When the vertex mapping is fixed, Pach and Wenger \cite{pach98} showed that all
planar graphs on $n$ vertices admit a polyline drawing with at most $O(n^2)$
bends. In the same paper, they showed that this result is optimal if the vertex
mapping is to points in convex position. It was left as an open problem whether
this lower bound applies to all fixed vertex mappings. We answer this open
problem in Theorem~\ref{thm:lower_bound} by extending the lower bound to all
fixed vertex mappings. That is, for any fixed vertex mapping, we show that
almost all planar graphs do not admit polyline drawings with $o(n^2)$ bends.

To prove our lower bound, we generalize the encoding techniques from
\cite{gordon12}, which were also limited to convex point sets. In particular, we
show how to encode a planar graph more efficiently when it admits a polyline
drawing with fewer bends (assuming a fixed vertex mapping). Very recently,
Francke and T{\'o}th \cite{francke14} considered this same encoding problem.
Their results gave an encoding bound of $O(\beta + n)$ bits when a polyline
drawing with $\beta$ bends is admitted. Our results significantly improve this
encoding bound to $n\lg(\beta/n) + O(n)$ bits.

Our encoding technique is described in Section~\ref{sec:encoding}. The approach
exploits a relation between a graph's separability and the numbers of bends with
which it can be drawn. This separability is formalized by
Lemma~\ref{lem:poly_separator} in Section~\ref{sec:separators}. The encoding
technique also requires an efficient way to encode the path traced by a
polyline, which is described in Section~\ref{sec:cld}. Finally, the lower bound
(Theorem~\ref{thm:lower_bound}) is presented in Section~\ref{sec:lower_bound}.

\section{Convex Layer Diagrams}
\label{sec:cld}

A \emph{convex layer diagram} is an embedding of $l$ convex polygons in
the plane such that the $i$-th convex polygon is strictly contained by the
$(i+1)$-th polygon. By convention, we refer to the outermost polygon as
layer-$1$ and the innermost polygon as layer-$l$. As with any simple polygon,
each layer defines a \emph{boundary} consisting of \emph{layer~edges} and
\emph{layer~vertices}. See Figure~\ref{fig:cld} for an example diagram.

\begin{figure}
\centering
\includegraphics[trim=0 1cm 0 0.5cm]{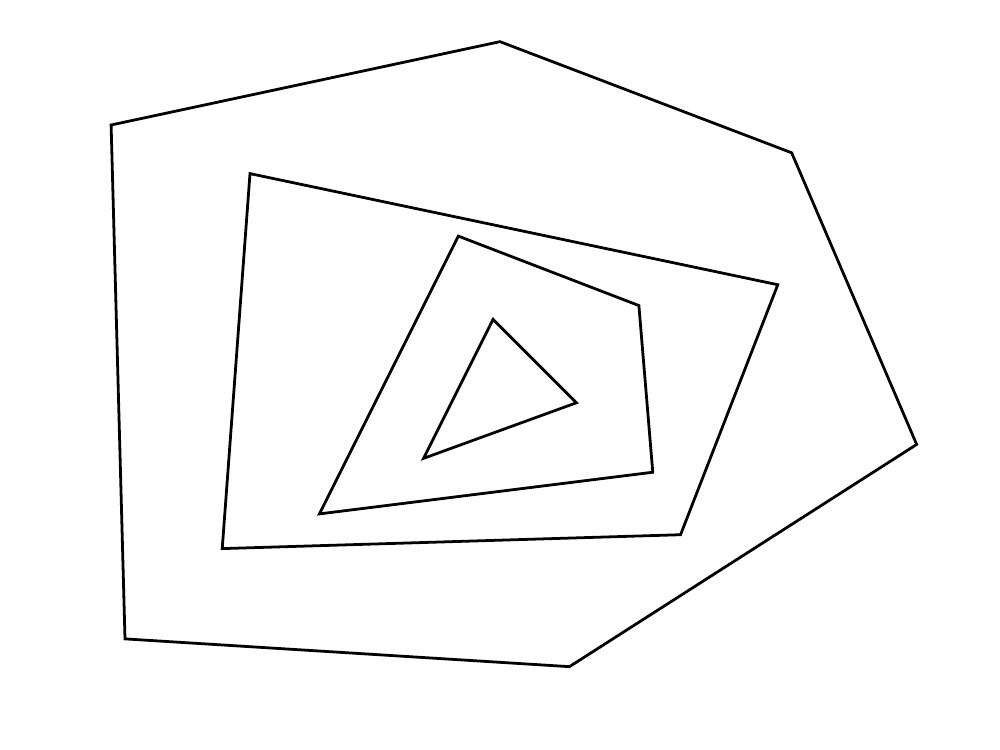}
\caption{A convex layer diagram with 4 layers.}
\label{fig:cld}
\end{figure}

Let $\Lambda$ be a fixed convex layer diagram drawn in $R^2$. A line segment is
said to \emph{cross} a layer edge if the two intersect at a point that is
interior to both the edge and the line segment. A line segment is said to cross
a layer boundary if it crosses one of its edges. The following lemma shows that
the layer edges in $\Lambda$ that a line segment crosses can be identified with
an ordered pair of \emph{support vertices}.

\begin{lemma}
\label{lem:support_vertices}
Suppose that a line segment crosses $k \geq 2$ layer boundaries in $\Lambda$.
Then, the layer edges crossed by this line segment can be identified by a pair
of support vertices $(v_i, v_j)$, each being incident to a crossed
layer edge.
\end{lemma}

\begin{proof}
Let $\Phi$ be the layer edges crossed by the line segment and let $\Psi$ be the
layer vertices incident to the edges in $\Phi$. Extend the line segment to form
a coincident line $\gamma$. We will show how to transform $\gamma$ via
translations and rotations so that it intersects with two vertices in $\Psi$,
while preserving its edge intersections in $\Phi$.

Start by perturbing $\gamma$ by an sufficiently small rotation and translation
to ensure that each vertex in $\Psi$ is at a distinct distance from $\gamma$
along the $x$-axis and that $\gamma$ is aligned with neither the $x$-axis nor
any edge in $\Phi$.

We can then translate $\gamma$ left until it first intersects a vertex $v_i$ in
$\Psi$. Call this new line $\gamma'$. Clearly, all edges in $\Phi$ intersect
$\gamma'$. Rotate $\gamma'$ counterclockwise about $v_i$ until it intersects
another vertex $v_j \neq v_i$, and call this new line $\gamma''$. By the same
argument $\gamma''$ must intersect all edges in $\Phi$. See
Figure~\ref{fig:support_vertices} for an example of this process.

\begin{figure}
\centering
\subfloat[][]{\includegraphics[trim=0 0.8cm 0 1cm]{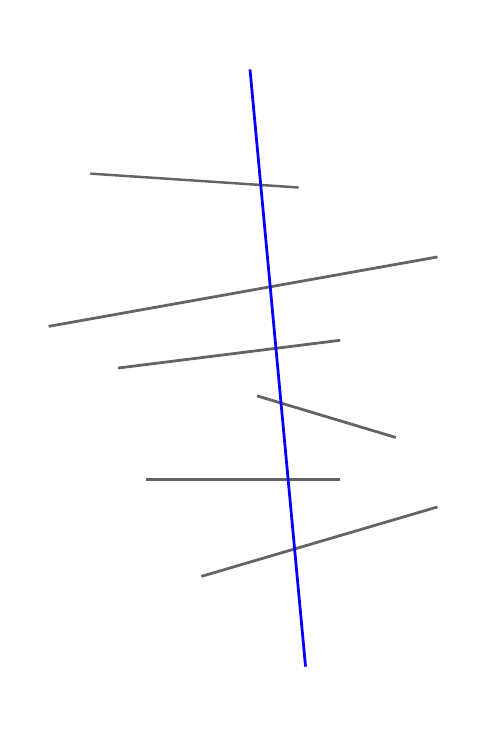}}
\subfloat[][]{\includegraphics[trim=0 0.8cm 0 1cm]{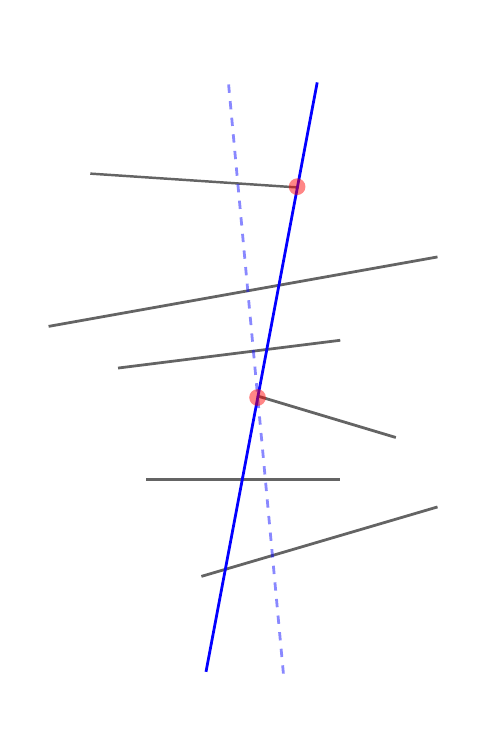}}
\caption{A line segment that crosses 6 layer edges is shown in (a). The
transformed line segment is shown in (b) with the identifying support vertices.}
\label{fig:support_vertices}
\end{figure}

The line $\gamma''$ is uniquely determined by the line through $v_i$ and $v_j$.
However, where $\gamma''$ intersects a layer vertex $u$, there is ambiguity
about which layer edge in $\Lambda$ incident $u$ had intersected with $\gamma$.
This ambiguity can be resolved by representing $\gamma''$ as the ordered pair
$(v_i, v_j)$. Indeed, from knowing that $v_i$ was used for rotation, we can
rotate $\gamma''$ back clockwise by an sufficiently small angle. The resulting
line can only intersect the one vertex $v_i$ in $\Psi$. By applying another
sufficiently small translation to the right, the edges $\Phi$ that $\gamma$
originally intersected can be unambiguously determined.
\end{proof}

\section{Simple Polygon Separators}
\label{sec:separators}

It is well known (\cite{pach98}, \cite{badent07}, \cite{gordon12}) that all
planar graphs admit polyline drawings using a fixed vertex mapping such that the
total number of bends is $O(n^2)$. We show that if a graph admits a polyline
drawing with fewer bends, then its vertices are correspondingly more separable.
This separability is formalized in Lemma~\ref{lem:poly_separator}.

To find the appropriate way of separating our graph $G$, we leverage the
following weighted simple cycle separator theorem.

\begin{lemma}[\cite{djidjev97}]
\label{lem:cycle_separator}
Let $G$ be a maximal planar graph with non-negative vertex weights adding to at
most 1. Then, the vertex set of $G$ can be partitioned into 3 sets $V_1, V_2, C$
such that neither $V_1$ nor $V_2$ has total vertex weight exceeding $2/3$, no edge
from $G$ joins a vertex in $V_1$ to a vertex in $V_2$, and the subgraph on $C$
is a simple cycle with at most $2\sqrt{n} + 7$ vertices.
\end{lemma}

Given a polyline drawing of a graph $G$, we can separate the vertices into two
parts by superimposing a simple polygon $\Omega$. See Figure~\ref{fig:separator}
for an example. The vertices that lie on or inside $\Omega$ define one part, and
the vertices that lie outside $\Omega$ define the other. Let $V_1$ be the
vertices on or inside $\Omega$, and let $V_2$ be the vertices outside $\Omega$.
If $|V_1|$ and $|V_2|$ are both at most $2/3n$ and at most $\alpha$ edges in $G$
join a vertex in $V_1$ to a vertex in $V_2$, then $\Omega$ defines an
\emph{$\alpha$-simple polygon separator}.

\begin{figure}
\centering
\includegraphics[trim=0 0.1cm 0 0.1cm]{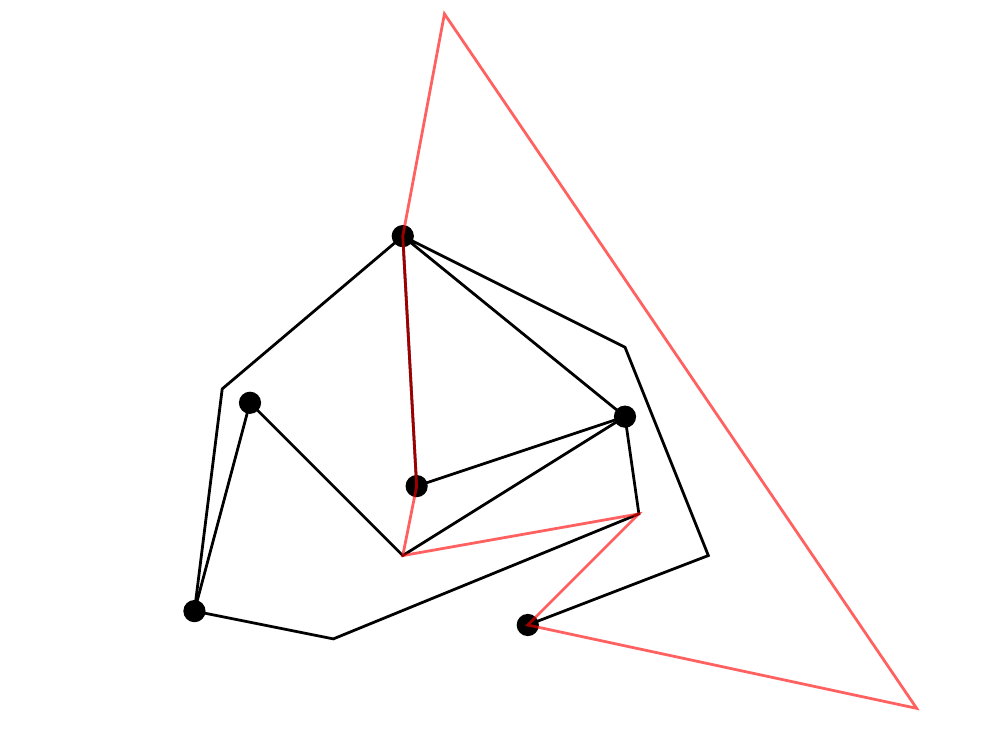}
\caption{A simple polygon separator superimposed on a graph. Two vertices lie
outside and four lie on or inside the separator. Three edges join vertices on
opposing sides of the separator.}
\label{fig:separator}
\end{figure}

\begin{lemma}
\label{lem:poly_separator}
Let $G$ be a planar graph on $n$ vertices with bounded vertex degree $d$.
Suppose that $G$ admits a polyline drawing with a total of $\beta$ bends. Then,
there is an $\alpha$-simple polygon separator having at most $r$ edges,
where $\alpha \leq dr$ and $r \leq 2\sqrt{n + \beta} + 10$.
\end{lemma}

\begin{proof}
Augment $G$ by adding $\beta + 3$ auxiliary vertices, one at each bend in its
polyline drawing and an additional 3 vertices to form a triangle bounding the
polyline drawing. We can further add edges to make a maximal planar graph $G'$
by triangulating this augmented polyline drawing, resulting in a drawing without
any bends (a \emph{straight-line embedding}). See
Figure~\ref{fig:triangulation} for an example.

Observe that $G'$ has $n + \beta + 3$ total vertices. If we weight each
auxiliary vertex $0$ and each original vertex $1/n$, it follows from
Theorem~\ref{lem:cycle_separator} that $G'$ contains a simple cycle separator
$C$ with at most $2\sqrt{n + \beta + 3} + 7$ edges. Moreover, the number of
original vertices lying inside this cycle and the number lying outside are both
at most $2/3n$.

\begin{figure}
\centering
\includegraphics[trim=0 0cm 0 0cm]{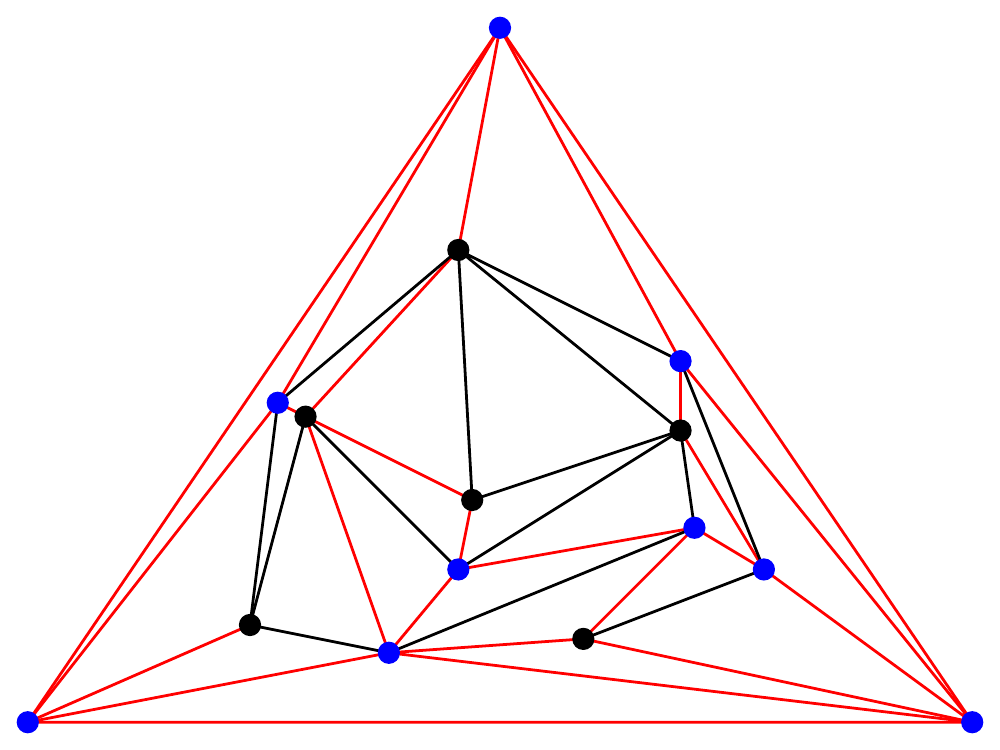}
\caption{The triangulation of a polyline drawing. Auxiliary vertices introduced
at bends and on the bounding triangle are in blue. The auxiliary edges are in
red.}
\label{fig:triangulation}
\end{figure}

Since our drawing of $G'$ had no bends, this cycle corresponds precisely to a
simple closed polygon $\Omega$ with $r$ edges, where $r$ is at most 
\[2\sqrt{n + \beta + 3}+ 7  \leq 2\sqrt{n + \beta} + 10. \]
The polyline of an edge in $G$ from a vertex on or inside $\Omega$ to a vertex
outside $\Omega$ had to intersect with one of the vertices in $C$. Since at most
$d$ edges in $G$ can intersect with any one vertex in $C$, it follows that
$\Omega$ is an $\alpha$-simple polygon separator where $\alpha \leq dr$.
\end{proof}

\section{Encoding Polyline Drawings}
\label{sec:encoding}

In this section, we describe a general method for encoding labeled planar
graphs. An encoding is simply a sequence of bits (our choice for unit of
information) from which the edges of the original graph can be unambiguously
determined. Any information that is said to be \emph{fixed} does not require
encoding as it is assumed to be available at decoding time. The bounds given for
our encoding algorithms are in terms of parameters that can be arbitrarily large
(e.g. the number of vertices or edges in the graph being encoded).

Our encoding bounds provide the means for proving the lower bound of
Theorem~\ref{thm:lower_bound}. However, a direct consequence of
Theorem~\ref{thm:encoding} resolves a conjecture from \cite{francke14} in the
positive. Namely, we show that there are at most $2^{n\lg(1 + \beta/n) + O(n)}$
$\beta$-bend polyline drawings on $n$ vertices whose mapping is fixed. 

The core component of this encoding technique is encapsulated by the following
lemma. The high-level idea of this lemma is as follows. First, we use
Lemma~\ref{lem:poly_separator} to partition a graph's vertices into 2 roughly
equal size parts $V_1,V_2$. This partition defines 3 subgraphs (edges between
$V_1$ vertices, edges between $V_2$ vertices, and edges between a $V_1$ and a
$V_2$ vertex). These 3 subgraphs are then encoded recursively. To complete the
encoding, we just need to show how to encode the partition so that the
original graph can be recovered from the 3 subgraphs. It turns out that the
partition given by Lemma~\ref{lem:poly_separator} can be encoded using only
$O(\sqrt{\beta + n}\lg(n/\sqrt{\beta + n}))$ bits. This additional information
is sufficiently small to give the desired bound.

\begin{lemma}
\label{lem:encoding}
Let $G$ be a planar graph on $n$ vertices where each vertex has degree at most
$2$. Suppose that $G$ admits a polyline drawing with a total of $\beta$ bends
for which the vertex mapping is fixed (i.e. the point at which each vertex is
drawn is fixed). Then, $G$ can be encoded using $n\lg(1 + \beta/n) + O(n)$ bits.
\end{lemma}

\begin{proof}
Define $\sigma = \beta + n$ (i.e. the total count of all vertices and bends),
and let $m$ be the number of edges in $G$ (note that $m \leq n$). We will prove
the stronger claim that $G$ can be encoded with $T(n,m,\sigma)$ bits, where
\[T(n,m,\sigma) \leq m\lg(\sigma/n) + 850n - 423\sqrtsiglg,\] 
by induction on $n$.

Since each vertex in $G$ has a degree at most $2$, the edges in $G$ can be
directed such that the outgoing degree of each vertex is at most $1$. Thus, if
we use $n$ bits to encode the vertices in $G$ with an outgoing edge, we can
trivially encode all edges in $G$ by specifying at most one vertex for each
edge. It follows that this encoding uses a total of $m\lg{n} + n$.  

Since $\sqrtsiglg \leq n$ always, we have that
\[427n \leq 850n - 423\sqrtsiglg.\]
Hence, this trivial encoding suffices for our base case as $m\lg{n} + n \leq
427n$ whenever $n \leq 2^{426}$. Suppose further that $\sigma \geq n^2/2^{426}$.
In this case, we have the inequality \[ m\lg{n} + n \leq m\lg(\sigma/n) + 427n\]
and thus the trivial encoding still satisfies the desired bound.

Thus, we can proceed by way of induction assuming that the claim holds for
smaller values of $n$ and that $n \geq 2^{426}$ and $\sigma < n^2/2^{426}$. Let
$\Pi$ be the points to which the vertices in $G$ are mapped in the polyline
drawing. We can form a convex layer diagram out of $\Pi$ as follows. Define
layer 1 by the convex hull of $\Pi$. Remove the points in $\Pi$ that lie on its
convex hull, and define layer 2 by the convex hull of the remaining points. We
can repeat this process to define the $l$ layers of a convex layer diagram
$\Lambda$.

The convex layer diagram $\Lambda$ depends only on the points $\Pi$ and is thus
fixed. We can furthermore assume a fixed iteration order over its layer edges,
starting from the outermost layer and iterating inward. To encode $G$, we will
leverage this fixed structure of $\Lambda$.

Since $G$ has bounded degree $2$, it follows by Lemma~\ref{lem:poly_separator}
that its polyline drawing has an $\alpha$-simple polygon separator $\Omega$ with
at most $r$ edges, where $\alpha \leq 2r$ and $r \leq 2\sqrtsig + 10$. To
simplify our encoding, we make modifications to $\Omega$. At each point in $\Pi$
that intersects with $\Omega$, we modify $\Omega$ in an $\epsilon$-neighborhood
around this point so that it no longer intersects with the point. This
modification can trivially be done by adding at most $2r$ edges (2 at each
intersecting point in $\Pi$ of which there are at most $r$). We also split any
edge in $\Omega$ that crosses a layer boundary more than once to ensure that
each edge crosses each layer boundary at most once. Doing so adds at most $r$
additional edges (each edge can cross a layer boundary at most twice by
convexity and thus at most once split per original edge suffices).

Thus, we can assume that $\Omega$ has at most $4r$ edges and partitions the
vertices in $G$ into two parts $V_1,V_2$, those mapping to points inside
$\Omega$ and those mapping to points outside $\Omega$. We can further assume
that $1/3n \leq |V_1| \leq |V_2| \leq 2/3n$ without loss of generality. We will
show how to encode this vertex partition using at most
$O(\sqrtsig\lg(n/\sqrtsig))$ bits.

Each $\Omega$-edge crosses a (possibly empty) set of layer edges in $\Lambda$.
If at most 1 layer edge is crossed by a given edge, define an auxiliary vertex
at this crossing. Otherwise, define an auxiliary vertex at the first and last
crossing. By Lemma~\ref{lem:support_vertices}, the entire set of layer edges
crossed by a given $\Omega$-edge can be recovered if we further encode its
support vertices. We introduce up to 2 additional auxiliary vertices at the
crossings with the edges incident to the support vertices.

The total number of auxiliary vertices is at most $16r$. Since there are at most
$n$ layer edges, it follows that we can encode the number of auxiliary vertices
that were added to each layer edge using $\lg{n + 16r \choose 16r} \leq
17r\lg(n/r)$ bits. We can define a cycle $C$ that joins these auxiliary
vertices in the order in which they intersect with $\Omega$. $C$ has fewer than
$n$ vertices since $16r \leq 32\sqrtsig + 160 < n$. Thus, if we adopt any fixed
convention for positioning the auxiliary vertices along the layer edges, we can
encode $C$ using $850(16r) = 13600r$ bits by the induction hypothesis. We can
further annotate the edges in $C$ with an additional $16r$ bits to encode which
vertices corresponded to support vertices. Thus, we have encoded the first and
last intersections of each edge in $\Omega$ as well as the support vertices
defining the interior intersections of this edge. We further know from $C$ the
structure of $\Omega$ between layer boundaries and have thus encoded how to
define a simple closed curve that is homotopic to $\Omega$ in $R^2 - \Pi$. This
encoding allows us to unambiguously define the vertex partition $V_1,V_2$, and
uses a total of $17r\lg(n/r) + 13616r$ bits. Since $\sigma < n^2/2^{426}$, it
follows that $\lg(n/\sqrt{\sigma}) > 213$. Using both the constraints that
$r \leq 2\sqrt{\sigma} + 10$ and $n \geq 2^{426}$, we can further show that
$\lg(n/r) \geq 426/2 - 2 = 212$. Thus, the number of bits used to encode
$V_1,V_2$ is at most
\[17r\lg(n/r) + 13616r \leq 82r\lg(n/r).\]

Using the vertex partition $V_1,V_2$, we can partition the edges of $G$ into one
of 3 subgraph $G_1$,$G_2$, and $G_3$. $G_1$ is defined by the edges in $G$
between vertices in $V_1$, $G_2$ is defined by the edges in $G$ between
vertices in $V_2$, and $G_3$ is defined by the edges between a vertex in $V_1$
and a vertex in $V_2$.

To complete the proof, we argue by way of induction as each of $G_1$, $G_2$, and
$G_3$ has fewer than $n$ vertices. Let $n_i$ and $m_i$ be the number of vertices
and edges, respectively, in $G_i$ for $i=1,2,3$. Similarly, define $\sigma_i =
\beta_i + n_i$, where $\beta_i$ is the number of bends in $G_i$ for $i=1,2,3$.
We can thus complete the encoding for a total of 
\[T(n,m,\sigma) \leq T(n_1,m_1,\sigma_1) + T(n_2,m_2,\sigma_2)
  + T(n_3,m_3,\sigma_3) + 82r\lg(n/r)\]
bits.  

Since $V_1,V_2$ were defined in terms of the $\alpha$-simple polygon separator
$\Omega$, it follows that $n_1 + n_2 = n$, $1/3 \leq n_1 \leq n_2 \leq 2/3n$, 
$n_3 \leq 4r$, and $m_1 + m_2 + m_3 = m$. Furthermore, we have the constraint
that $\beta_1 + \beta_2 + \beta_3 \leq \beta$. Subject to these constraints, our
encoding size is thus
\begin{align*}
T(n,m,\sigma) &\leq m_1\lg(\sigma_1/n_1) + 850n_1 - 423\sqrtsiglga \\
              {} &+ m_2\lg(\sigma_2/n_2) + 850n_2 - 423\sqrtsiglgb \\
              {} &+ m_3\lg(\sigma_3/n_3) + 850n_3 - 423\sqrtsiglgc \\
              {} &+ 82r\lg(n/r)
\end{align*}
by induction. The remainder of the proof is to bound this expression. By using
elementary calculus, one can show that this bound is largest when $m_3=0$
(Alternatively, one could redefine $G_1$ to include the edges of $G_3$, removing
the need to encode $G_3$ by adding only an $850r$ term). Thus, our encoding size
satisfies the inequality
\begin{align*}
T(n,m,\sigma) &\leq m_1\lg(\sigma_1/n_1) + 850n_1 - 423\sqrtsiglga \\
              {} &+ m_2\lg(\sigma_2/n_2) + 850n_2 - 423\sqrtsiglgb \\
              {} &+ 82r\lg(n/r)
\end{align*}
where $m_1 + m_2 = m$ and $\sigma_1 + \sigma_2 \leq \sigma$. Since each vertex
in $G_1$ and $G_2$ has degree at most 2, it follows that $m_1 \leq n_1$ and
$m_2 \leq n_2$. Thus, by our constraints on $n_1$ and $n_2$, we also have that 
$1/3m \leq m_1,m_2 \leq 2/3m$. Again, we can show using elementary calculus that
this constrained inequality reaches its maximum when
\[\frac{m_1}{m},\frac{n_1}{n},\frac{\sigma_1}{\sigma} = 1/3
  \quad \text{and} \quad
  \frac{m_2}{m},\frac{n_2}{n},\frac{\sigma_2}{\sigma} = 2/3\]
(see \cite{gordon12} for more details on this derivation). Thus, it follows that
\begin{align*}
T(n,m,\sigma) &\leq m\lg(\sigma/n) + 850n - 423\sqrtsiglg \\
              {} &+ 82r\lg(n/r) + 423\sqrtsig\lg{\sqrt{9/2}} \\
              {} &- (\sqrt{1/3} + \sqrt{2/3} - 1)423\sqrtsiglg
\end{align*}
and since $82r\lg(n/r) + 423\sqrtsig\lg{\sqrt{9/2}} \leq 165\sqrtsiglg$, it
follows that the last two lines cancel and
\begin{align*}
T(n,m,\sigma) &\leq m\lg(\sigma/n) + 850n_1 - 423\sqrtsiglg
\end{align*}
completing the proof.
\end{proof}

We can in fact generalize our encoding to apply to any planar graph by a simple
reduction to Lemma~\ref{lem:encoding}. This generalization gives the following
theorem.

\begin{theorem}
\label{thm:encoding}
Let $G$ be a planar graph that admits a polyline drawing with a total of $\beta$
bends under a fixed vertex mapping. Then, $G$ can be encoded using
$n\lg(1 + \beta/n) + O(n)$ bits.
\end{theorem}

\begin{proof}
Let's assume that $G$ is connected, and let $T$ be a spanning tree of $G$.
Suppose that we traverse the edges bounding the unique face of $T$. Since $T$ is
a tree, we will traverse each edge exactly twice (both sides must be on this
unique face). If we start the traversal along an arbitrary edge, and stop before
traversing it a third time, we will visit a sequence of $2n$ vertices
$S=v_1,v_2,\dots,v_{2n}$.

Let $v_i$ be a vertex in $T$ and let $j$ be the first position at which it
occurs in $S$ from the left. We define $\suc{v_i}$ as the leftmost occurring
vertex in $S$ whose first occurrence is at a position beyond $j$. We can then
define a sequence of vertices $S'=v_i,\suc{v_i},\suc{\suc{v_i}},\dots$,
terminating when a vertex has no successor. This sequence of vertices
corresponds to a path $P$ passing through each vertex in $T$ exactly once. See
Figure~\ref{fig:traversal} for an example of this construction.

\begin{figure}
\centering
\includegraphics[trim=0 0.5cm 0 0.5cm]{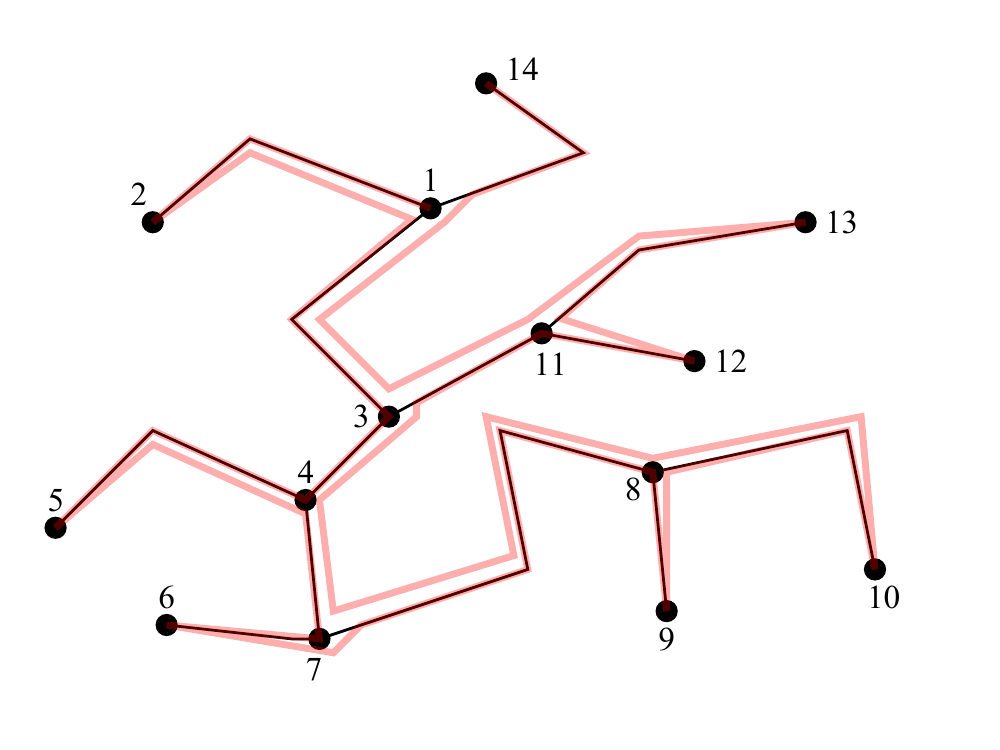}
\caption{The path (in red) defined by the traversal of the unique face of a tree
(in black).}
\label{fig:traversal}
\end{figure}

From the polyline drawing of $G$ we can immediately construct a polyline drawing
of $T$ having at most $\beta$ bends. We can thus construct a polyline drawing of
$P$ using $2\beta + O(n)$ bends since the edges in $P$ correspond to a traversal
along face of $T$. It follows from Lemma~\ref{lem:encoding} that $P$ can be
encoded using
\[ n\lg(1 + (2\beta + O(n))/n) + O(n) = n\lg(1 + \beta/n) + O(n) \]
bits. From $P$ it is easy to recover $T$ with only an additional $2n$ bits.
Indeed, the vertex sequence in $P$ (that is, $S'$) defines a subsequence of the
face of $T$, the rest of which can be encoded by saying which vertices in $P$ to
traverse back on.

Having encoded $T$, we only need to encode the remaining edges in $G$ that are
not in $T$. The number of possible arrangements of these edges corresponds to
the well studied Catalan numbers (i.e. they can be represented with a
parenthesization of length $O(n)$). It follows that $G$ can be encoded with only
$O(n)$ bits more than the encoding of $P$ for a total of
$n\lg(1 + \beta/n) + O(n)$ bits.

Finally, if $G$ in not connected, we can trivially make $G$ connected by adding
edges without introducing more than $\beta$ additional bends, giving the same
bound.
\end{proof}

\section{Lower Bound on the Number of Bends}
\label{sec:lower_bound}

Let $\Pi$ be a set of $n$ points in $R^2$, and let $\pi : V \to \Pi$ be a fixed
vertex mapping. The open question from \cite{pach98} asked whether there exists
a $\pi$ for which all planar graphs admit a polyline drawing with $o(n^2)$ total
bends. The following theorem shows that this cannot be accomplished and that,
for any $\pi$, the set of planar graphs requiring $\Omega(n^2)$ bends becomes
dense as $n \to \infty$.

\begin{theorem}
\label{thm:lower_bound}
Let $\pi : V \to \Pi$ be a fixed vertex mapping, and let $G = (V, E)$ be a
planar graph sampled uniformly at random from the set of all planar graphs on
$n$ vertices. Then, with high probability, all polyline drawings of $G$ using
the vertex mapping $\pi$ have $\Omega(n^2)$ bends.
\end{theorem}

\begin{proof}
Suppose that $G$ is a planar graph sampled uniformly at random from the set of
labeled planar graphs on $n$ vertices. The number of such graphs is known to be
$\Theta(n^{-7/2}\gamma^nn!)$, where $\gamma \approx 27.22687$. Thus, for $n$
sufficiently large, it follows that at least $\lg{n!} - \Delta$ bits are
required to encode $G$ with probability at least $1 - 2^{-\Delta}$.

Suppose that $G$ admits a polyline drawing using the vertex mapping $\pi$ having
a total of $\beta$ bends. It follows from Theorem~\ref{thm:encoding} that $G$
can be encoded with
\[n\lg(1 + \beta/n) + O(n) \]
bits. Thus, the inequality
\[n\lg(1 + \beta/n) + O(n) \geq \lg{n!} - \Delta\]
must hold with probability at least $1 - 2^{-\Delta}$. We can use Stirling's
approximation to show that $\lg{n!} \geq n\lg{n} - n$ and thus
\[n\lg(1 + \beta/n) + O(n) + \Delta \geq n\lg{n} \]
or equivalently,
\[n\lg\left(\frac{n + \beta}{n}2^{O(1) + \Delta/n}\right) \geq n\lg{n} \]
holds with probability at least $1 - 2^{-\Delta}$. If we divide through by $n$
and exponentiate both sides, it follows that
\[\frac{n + \beta}{n}2^{O(1) + \Delta/n} \geq n \]
or equivalently,
\[\beta \geq \frac{n^2}{2^{O(1) + \Delta/n}} - n\]
with probability at least $1 - 2^{-\Delta}$. In particular, we can choose
$\Delta$ to be $n$, which shows that $\beta$ is $\Omega(n^2)$ with probability
at least $1 - 2^{-n}$.
\end{proof}

\section{Open Problems}

We have shown that no fixed vertex mapping admits polyline drawings with
$o(n^2)$ total bends for all planar graphs. Moreover, the same result applies
for paths, trees, outerplanar graphs, or any subset of planar graphs with at
least $n!/2^{O(n)}$ elements on $n$ labeled vertices. 

On the other hand, using the techniques from \cite{gordon12} we can always
construct a polyline drawing such that each edge bends at most $3n$ times each
(in fact, $2n$ times each for a sufficiently random planar graph). Since these
results are tight up to constant factors, we ask the following questions:
\begin{enumerate}
\item Can the constant 3 be improved? That is, assuming a fixed vertex mapping,
do all planar graph admit polyline drawings having fewer than $3n$ bends? Do
some fixed vertex mappings give provably better constant factors than others?
\item How tight of a constant can be shown in the lower bound? Using our
approach, this lower bound constant directly maps to the constant of the
$O(n)$ term in the encoding of Theorem~\ref{thm:encoding}. To what extent can
the $O(n)$ term constant be reduced?
\item Can our encoding technique be used to prove other lower bounds in graph
drawing? Are there other applications of this encoding technique outside of
graph drawing?
\end{enumerate}

\bibliographystyle{splncs}
\bibliography{references}

\end{document}